\pdfoutput=1
\RequirePackage{ifpdf}
\ifpdf 
\documentclass[pdftex]{sigma}
\else
\documentclass{sigma}
\fi

\numberwithin{equation}{section}

\newtheorem{Theorem}{Theorem}[section]
\newtheorem{Corollary}[Theorem]{Corollary}
\newtheorem{Lemma}[Theorem]{Lemma}
\newtheorem{Proposition}[Theorem]{Proposition}
 { \theoremstyle{definition}

\newtheorem{Example}[Theorem]{Example}
\newtheorem{Remark}[Theorem]{Remark} }

\usepackage{xcolor}

\newcommand{\ii}{\text{i}}

\newcommand{\dd}[1]{\text{d}#1 }

\newcommand{\be}{\begin{equation}}
	\newcommand{\ee}{\end{equation}}
\newcommand{\bes}{\begin{equation*}}
	\newcommand{\ees}{\end{equation*}}

\newcommand{\ttype}{\mathfrak{t}}

\begin{document}
\allowdisplaybreaks

\newcommand{\arXivNumber}{2004.06035}

\renewcommand{\PaperNumber}{102}

\FirstPageHeading

\ShortArticleName{Triangle Groups: Automorphic Forms and Nonlinear Differential Equations}

\ArticleName{Triangle Groups: Automorphic Forms\\ and Nonlinear Differential Equations}

\Author{Sujay K.~ASHOK~$^\dag$, Dileep P.~JATKAR~$^\ddag$ and Madhusudhan RAMAN~$^\S$}

\AuthorNameForHeading{S.K.~Ashok, D.P.~Jatkar and M.~Raman}

\Address{$^\dag$~Institute of Mathematical Sciences, Homi Bhabha National Institute (HBNI),\\
\hphantom{$^\dag$}~IV Cross Road, C.I.T.~Campus, Taramani, Chennai 600 113, India}
\EmailD{\href{mailto:sashok@imsc.res.in}{sashok@imsc.res.in}}

\Address{$^\ddag$~Harish-Chandra Research Institute, Homi Bhabha National Institute (HBNI),\\
\hphantom{$^\ddag$}~Chhatnag Road, Jhunsi, Allahabad 211 019, India}
\EmailD{\href{mailto:dileep@hri.res.in}{dileep@hri.res.in}}

\Address{$^\S$~Department of Theoretical Physics, Tata Institute of Fundamental Research,\\
\hphantom{$^\S$}~Homi Bhabha Road, Navy Nagar, Colaba, Mumbai 400 005, India}
\EmailD{\href{mailto:madhur@theory.tifr.res.in}{madhur@theory.tifr.res.in}}

\ArticleDates{Received April 21, 2020, in~final form October 05, 2020; Published online October 11, 2020}

\Abstract{We study the relations governing the ring of quasiautomorphic forms associated to triangle groups with a~single cusp, thereby extending our earlier results on Hecke groups. The Eisenstein series associated to these triangle groups are shown to satisfy Ramanujan-like identities. These identities in turn allow us to associate a~nonlinear differential equation to each triangle group. We show that they are solved by the quasiautomorphic \mbox{weight-2} Eisenstein series associated to the triangle group and its orbit under the group action. We conclude by discussing the Painlev\'e property of these nonlinear differential equations.}


\Classification{34M55; 11F12; 33E30}

\section{Introduction and discussion}
Triangle groups are an~infinite family of discrete subgroups of~${\rm PSL}(2,\mathbb{R}) $, the group of orienta\-tion-preserving isometries of the upper-half plane. Abstractly, these groups are characterised
by an~ordered triple of positive integers $ \lbrace m_i \rbrace_{i=1}^{3} $ and generated by the letters $ \lbrace g_i \rbrace_{i=1}^{3} $, and circumscribed by the relations
\begin{gather*}
g_i^{m_i} = 1 \qquad \text{for} \ i \in \lbrace 1,2,3 \rbrace \qquad \text{and} \qquad g_1 g_2 g_3 = 1.
\end{gather*}
The theory of automorphic functions and forms associated to triangle groups has been detailed in~\cite{Doran:2013npa}. Triangle groups with arithmetic properties have appeared in~the study of string theories~\cite{Cheng:2008kt, Jatkar:2005bh, Persson:2015jka}, supersymmetric gauge theories~\cite{Ashok:2016oyh}, the study of knots~\cite{Milnor, tsanov2010triangle, Zagier2010}, and simple quantum mechanical systems~\cite{Basar:2017hpr, Raman:2020sgw} and further motivates our investigations. The present work can be thought of as an~extension of our results for Hecke groups obtained in~\cite{Ashok:2018myo} to all triangle group with a~single cusp.

Differential equations satisfied by modular~\cite{huber2011differential,matsuda2016} and quasi-modular~\cite{ablowitz2006integrable,bureau1987systemes,ramamani1989} objects have been the subject of sustained interest. The primary objects of study in~this paper are an~infinite class of nonlinear differential equations naturally associated to triangle groups. In Section~\ref{sec:Ramanujan}, we review the construction of quasi-automorphic forms associated to arbitrary triangle groups as solutions to the generalised Halphen system, a~system of first-order differential equations, following~\cite{Doran:2013npa}. The automorphic forms associated to these triangle groups are then found to satisfy a~system of first-order differential equations whose structure is strongly reminiscent of the Ramanujan identities governing the ring of (quasi-)modular forms of ${\rm PSL}(2,\mathbb{Z})$.

The construction of automorphic forms also exhibits the existence of certain ring relations~-- that is, a~set of relations that imply that the ring of automorphic forms associated to a~triangle group is not freely generated. These ring relations play a~central role in~this note. Indeed, in~Section~\ref{MaierForTriangles} we use both the Ramanujan identities and the ring relations to derive a~non-linear differential equation previously derived in~\cite{Maier}. Notably, our derivation does not rely on any explicit Fourier expansions or require the use of computer algebra.

In the final section, we present the derivation of higher-order differential equations associated to each triangle group via differential elimination on the Ramanujan identities. These equations are natural analogues of the Chazy equation associated to~${\rm PSL} (2,\mathbb{Z}) $. We then prove that these differential equations possess the Painlev\'e property, which is to say that the only movable singularities of these differential equations are poles.

In his analysis of second-order ordinary differential equations, Painlev\'e restricted his attention to only those equations whose movable singularities were poles. All these equations turned out to be integrable. Based on this, the authors of~\cite{Ablowitz:1980ca} devised the Painlev\'e test, a~set of criteria that, if satisfied by a~differential equation, establishes the Painlev\'e property. This test works perfectly for all soliton equations but requires a~more careful stability analysis for those equations which possess negative resonances, along the lines of~\cite{Conte:1993pp}. The nonlinear differential equations we study in~this note generically possess negative resonances. In some cases, as we will show, this is no obstruction to the demonstration of the Painlev\'e property.

A proof of the Painlev\'e property is a~prerequisite to the demonstration of integrability. Indeed, the higher-order Chazy equations uncovered in~\cite{Ashok:2018myo} already possess tantalising hints that these equations may be integrable. Consider as an~example the Chazy equation corresponding to the Hecke group~${\rm H}(4) $ in~\cite[Theorem~4.3]{Ashok:2018myo}, written in~terms of the Ramanujan--Serre deri\-va\-tive~$\text{D}$ and a~weight-$4$ modular form~$Z$:
\begin{gather*}
\mathrm{D}^{3} Z+6 Z \mathrm{D} Z=0.
\end{gather*}
This presentation manifests a~striking resemblance to the stationary Korteweg--de Vries equation:
\begin{gather*}
u''' + 6 u u' = 0.
\end{gather*}
Given the uniform structure of the Ramanujan identities for triangle groups we write down, and the higher-order Chazy and Maier equations that are consequent upon them, one is tempted to look for a~framework, perhaps even a~hierarchy, within which these differential equations may be organised. The results of this note, in~particular investigations into the Painlev\'e property, are an~important step in~this direction.

\section{Triangle groups and Ramanujan identities}\label{sec:Ramanujan}

We begin with a~brief review of triangle groups and the associated generalised Halphen system. We will then review the main theorems of~\cite{Doran:2013npa} regarding the construction of Eisenstein series in~terms of the solutions to the Halphen system. Starting from these definitions, we then derive the Ramanujan relations that are valid for any triangle group with a~single cusp.

\subsection{Triangle groups}
A triangle group is identified by its \emph{type}, denoted by $\ttype$ and consisting of a~triple of integers $ 2\leq m_1 \leq m_2 \leq m_3 \leq \infty $. The integers that make up a~triple define a~triple of angles $ (\pi/m_1, \pi/m_2, \pi/m_3) $, which are taken to be the angles subtended by the edges of a~triangle in~the hyperbolic plane. The group of reflections of this triangle across its sides generates a~tessellation of the hyperbolic plane, and is called a~triangle group.

Triangle groups are discrete subgroups of~${\rm PSL}(2,\mathbb{R}) $, the group of isometries of the upper-half plane $ \mathbb{H} $, and we study those whose fundamental domains have finite hyperbolic area.\footnote{As a~result, the following types are not considered: (i) $(2, 2, m)$ for all $ m \leq \infty$, (ii)~$(2, 3, n) $ for $ n \leq 6$, and (iii)~the types $ (2, 4, 4) $ and $ (3, 3, 3) $.} We will further restrict our attention to those triangle groups with a~single cusp~-- meaning the angle subtended by one pair of edges of the aforementioned triangle is zero -- and without loss of generality this choice will be implemented by setting $ m_3 = \infty $.\footnote{For a~more algebraic characterisation, these triangle groups are isomorphic to the free product group $ \mathbb{Z}_{m_{1}} * \mathbb{Z}_{m_{2}} $, where the notation $ \mathbb{Z}_k $ denotes a~cyclic group of order $ k $.} In what follows, when we refer to triangle groups we will always mean triangle groups with a~single cusp, characterised by a~pair of integers so its type $ \ttype = (m_1, m_2,\infty) $ such that $ 2 \leq m_1 \leq m_2 < \infty$. For example, the modular group ${\rm PSL}(2,\mathbb{Z}) $ is a~triangle group of type $ (2,3,\infty) $, and the Hecke groups ${\rm H}(m) $ are of type $ (2,m,\infty) $. Finally, throughout this paper, the subscript $ \mathfrak{t} $ will indicate that the object in~question is associated to any specific choice of triangle group.

We now move on to discuss the construction of Eisenstein series associated to these triangle groups. This will be done following~\cite{Doran:2013npa}, who construct these automorphic forms out of solutions to a~generalised Halphen system.

\subsection{Generalized Halphen systems and Eisenstein series}

The generalized Halphen system is a~set of coupled ordinary first-order differential equations for three variables $ \lbrace t_{k, \ttype}(\tau) \rbrace_{k=1}^{3} $ given as follows:
\begin{gather}
t_{1,\ttype}' =(a-1)(t_{1,\ttype}t_{2,\ttype}+t_{1,\ttype}t_{3,\ttype}-t_{2,\ttype}t_{3,\ttype})+(b+c-1) t_{1,\ttype}^2,\nonumber\\
t_{2,\ttype}' =(b-1)(t_{2,\ttype}t_{3,\ttype}+t_{2,\ttype}t_{1,\ttype}-t_{1,\ttype}t_{3,\ttype})+(a+c-1) t_{2,\ttype}^2,\nonumber\\
t_{3,\ttype}' =(c-1)(t_{3,\ttype}t_{1,\ttype}+t_{3,\ttype}t_{2,\ttype}-t_{1,\ttype}t_{2,\ttype})+(a+b-1) t_{3,\ttype}^2.\label{eq:GenHalphen}
\end{gather}
Here, the parameters $\lbrace a,b,c \rbrace$ are specified by the type $\ttype$ of the triangle group as
\begin{gather}
a =\frac{1}{2}\left(1-\frac{1}{m_1}+\frac{1}{m_2} \right)\!,\qquad
b =\frac{1}{2}\left(1-\frac{1}{m_1}-\frac{1}{m_2}\right)\!,\qquad
c = \frac{1}{2}\left(1+\frac{1}{m_1}-\frac{1}{m_2} \right)\!,\label{eq:abcHalphen}
\end{gather}
and the accent $ ' $ denotes the following derivative:
\begin{gather*}
h' := \frac{1}{2\pi\ii} \frac{\dd{}}{\dd{\tau}} h(\tau),
\end{gather*}
where $ \tau $ is a~coordinate on the upper-half plane $ \mathbb{H} $.

The solutions to the generalized Halphen system can be obtained explicitly in~terms of hyper\-geometric functions \cite[Theorem~3(i)]{Doran:2013npa}. We relegate a~discussion of these solutions to Appen\-dix~\ref{app:Fourier}~-- notably, our conventions differ slightly from theirs~-- and instead proceed to the construction of Eisenstein series. We start with the following theorem due to \cite[see p.~707 and~Theorem~4(iv)]{Doran:2013npa}:

\begin{Theorem}\label{theorem21}The algebra of holomorphic automorphic forms of a~triangle group $ \ttype $ is generated by two branches of Eisenstein series: $ \big\lbrace E_{2k,\mathfrak{t}}^{(1)}\big\rbrace_{k=2}^{m_1} $ and $ \big\lbrace E_{2k,\mathfrak{t}}^{(2)}\big\rbrace_{k=2}^{m_2} $. For each admissible $ k $, the corresponding Eisenstein series has an~automorphic weight $ 2k $ under the triangle group. These Eisenstein series are constructed out of solutions to the generalised Halphen system as follows: define the linear combinations
\begin{gather}\label{eq:xydef}
\mathsf{x_{\ttype}}= t_{1,\ttype}-t_{2,\ttype} \qquad \text{and}\qquad
\mathsf{y_{\ttype}}= t_{3,\ttype}-t_{2,\ttype},
\end{gather}
in terms of which the Eisenstein series are
\begin{alignat}{3}
& E_{2k,\ttype}^{(1)} = \mathsf{x}_{\ttype} \mathsf{y}_{\ttype}^{k-1} \qquad&& \text{for}\ 2\le k \le m_1 ,& \\
&E_{2k,\ttype}^{(2)} = \mathsf{x}_{\ttype}^{k-1} \mathsf{y}_{\ttype}\qquad&&\text{for}\ 2\le k \le m_2 .&\label{eq:E2kxy}
\end{alignat}
\end{Theorem}

The above construction of the Eisenstein series in~terms of solutions to the generalised Halphen system straightforwardly implies the following results.

\begin{Corollary}The algebra of automorphic forms is not freely generated. For both branches of Eisenstein series, i.e., for $ \ell \in \lbrace 1,2 \rbrace $ we have the following relations:
	\begin{equation}	\label{eq:notFreelyGenerated}
	E_{2p, \ttype}^{(\ell)}E_{2(k-p+1), \ttype}^{(\ell)} = E_{2p',\ttype}^{(\ell)}E_{2(k-p'+1),\ttype}^{(\ell)} ,
	\end{equation}
	for any integers $2 \leq p \leq k-1 $ and $ 2 \leq p' \leq k-1$. Additionally, we have the following relations binding the two branches of Eisenstein series:
	\begin{equation}
	\label{newringrelation}
	E_{2k,\ttype}^{(1)} E_{2k,\ttype}^{(2)} = (E_{4,\ttype})^{k} .
	\end{equation}
\end{Corollary}
Collectively, we will refer to the above identities as ring relations. They will be used extensively in~the following sections.

\begin{Remark}
	Associated to every triangle group, there exists a~unique weight-$4$ Eisenstein series, obtained as the case $k = 2$ of~\eqref{eq:E2kxy}. Since this automorphic form may be thought of as belonging to either branch of Eisenstein series, we denote it without any reference to these branches, as~$E_{4,\ttype}$.
\end{Remark}

The triangle groups also come equipped with a~quasi-automorphic weight $2$ Eisenstein series~$ E_{2,\ttype} $, which is also given in~terms of a~linear combination of the generalised Halphen variables~$t_{i,\ttype}$ \cite[Theorem 4(iii)]{Doran:2013npa}.
\begin{Theorem}
\label{theorem22}
	For any triangle group with type $\ttype$, we can associate a~holomorphic, quasi-automorphic weight-$2$ Eisenstein series~$E_{2,\ttype}$. In terms of solutions to a~generalised Halphen sys\-tem, it is given by
	\begin{gather}
	\label{eq:E2xy}
	E_{2,\ttype} = \frac{1}{m_1+m_2-m_1m_2} \big( 2m_1 \mathsf{x}_{\ttype} + 2m_2 \mathsf{y}_{\ttype} + (m_1+m_2+m_1m_2) t_{2,\ttype} \big).
	\end{gather}
\end{Theorem}
This choice of linear combination is tied to the choice of normalisation for solutions to the generalised Halphen system and ensures that at $ \ii\infty $, $ E_{2,\ttype} $ is unity -- see Appendix~\ref{app:Fourier}.

With the introduction of the Eisenstein series that generate the ring of holomorphic, quasi-automophic forms of triangle groups $ \ttype $ now complete, we are in~a position to state our first result. It generalizes our previous results on Ramanujan identities for Hecke groups~\cite{Ashok:2018myo, Raman:2018owg} to an~arbitrary triangle group with cusp.

\begin{Lemma}
	The Eisenstein series associated to a~triangle group $ \ttype $, as defined in Theorems~{\rm \ref{theorem21}} and~{\rm \ref{theorem22}}, and for $ k > 2 $, obey the following identities:
\begin{gather}
E_{2,\ttype}' = \frac{(m_1m_2-m_1-m_2)}{2m_1m_2} \left( E_{2,\ttype}^2 - E_{4,\ttype} \right)\!,\nonumber\\
E_{4,\ttype}' = 2\frac{(m_1m_2-m_1-m_2)}{m_1m_2} E_{2,\ttype} E_{4,\ttype} -\frac{m_1-2}{m_1} E_{6,\ttype}^{(1)} -\frac{m_2-2}{m_2} E_{6,\ttype}^{(2)},\nonumber\\
E_{2k,\ttype}^{{(1)}'} = \frac{k(m_1m_2\!-\!m_1\!-\!m_2)}{m_1m_2} E_{2,\ttype}E_{2k,\ttype}^{(1)} \!-\!\frac{(km_2\!-\!k\!-\!m_2)}{m_2} E_{4,\ttype}E_{2k\!-\!2,\ttype}^{(1)}\!-\!\frac{(m_1\!-\!k)}{m_1} E_{2k+2,\ttype}^{(1)},\nonumber\\
E_{2k,\ttype}^{{(2)}'} =\frac{k(m_1m_2\!-\!m_1\!-\!m_2)}{m_1m_2} E_{2,\ttype}E_{2k,\ttype}^{(2)} \!-\!\frac{(km_1\!-\!k\!-\!m_1)}{m_1} E_{4,\ttype}E_{2k\!-\!2,\ttype}^{(2)} \!-\!\frac{(m_2\!-\!k)}{m_2} E_{2k+2,\ttype}^{(2)}.\label{eq:RamanujanIdGeneral}
\end{gather}
\end{Lemma}

\begin{proof}We perform a~linear transformation from the generalized Halphen variables $\lbrace t_{i,\ttype} \rbrace_{i=1}^{3}$ to the variables $(E_{2,\ttype}, \mathsf{x}_{\ttype}, \mathsf{y}_{\ttype})$ using \eqref{eq:xydef} and \eqref{eq:E2xy}. The generalised Halphen system is then equivalent to the following system of ordinary first-order differential equations:
\begin{gather}
E_{2,\ttype}' = \frac{(m_1m_2-m_1-m_2)}{2m_1m_2} \left( E_{2,\ttype}^2 - E_{4,\ttype} \right)\!,\nonumber\\
\mathsf{x}_{\ttype}' =\frac{(m_1m_2-m_1-m_2)}{m_1m_2}E_{2,\ttype} \mathsf{x}_{\ttype} + \frac{\mathsf{x}_{\ttype}^2}{m_2} -\left(1-\frac{1}{m_1}\right) E_{4,\ttype},\nonumber\\
\mathsf{y}_{\ttype}' =\frac{(m_1m_2-m_1-m_2)}{m_1m_2}E_{2,\ttype} \mathsf{y}_{\ttype} + \frac{\mathsf{y}_{\ttype}^2}{m_1} -\left(1-\frac{1}{m_2}\right) E_{4,\ttype} .\label{eq:xydbydtau}
\end{gather}
The first equation in~\eqref{eq:xydbydtau} proves the first of the Ramanujan identities. For $E_{4,\ttype}$, we consider the derivative of~\eqref{eq:E2kxy} for $ k = 2 $ to get
\begin{gather*}
E_{4,\ttype}' = \left[ \frac{(m_1m_2-m_1-m_2)}{m_1m_2}E_{2,\ttype} \mathsf{x}_{\ttype} + \frac{\mathsf{x}_{\ttype}^2}{m_2} -\left(1-\frac{1}{m_1}\right) \mathsf{x}_{\ttype} \mathsf{y}_{\ttype}\right] \mathsf{y}_{\ttype}
\\ \hphantom{E_{4,\ttype}' =}
{}+ \mathsf{x}_{\ttype}\left[ \frac{(m_1m_2-m_1-m_2)}{m_1m_2}E_{2,\ttype} \mathsf{y}_{\ttype} + \frac{\mathsf{y}_{\ttype}^2}{m_1} -\left(1-\frac{1}{m_2}\right) \mathsf{x}_{\ttype} \mathsf{y}_{\ttype}\right]
\\ \hphantom{E_{4,\ttype}'}
{}= 2\frac{(m_1m_2-m_1-m_2)}{m_1m_2} E_{2,\ttype} E_{4,\ttype} -\frac{m_1-2}{m_1} E_{6,\ttype}^{(1)} -\frac{m_2-2}{m_2} E_{6,\ttype}^{(2)} .
\end{gather*}
Similarly, for the relations governing either branch of Eisenstein series $ E_{2k,\ttype}^{(\ell)} $ of weight $2k>4$, differentiate \eqref{eq:E2kxy} and use \eqref{eq:xydbydtau}, then resolve all products of~$ \mathsf{x_{\ttype}} $ and $ \mathsf{y_{\ttype}} $ into Eisenstein series according to~\eqref{eq:E2kxy}.
\end{proof}

Owing to their similarities with relations satisfied by the Eisenstein series associated to $\mathrm{PSL}(2,\mathbb{Z}) $, these identities will henceforth be referred to as the Ramanujan identities associated to the triangle group~$ \ttype $.

\section{Maier equations for triangle groups}\label{MaierForTriangles}

The Ramanujan identities are a~set of coupled first order differential equations that involve all the Eisenstein series. However as we have seen, there are non-trivial ring relations that relate the higher weight Eisenstein series. We now show how it is possible to use the ring relations and deduce from the Ramanujan rela\-tions a~higher order non-linear differential equation for~every triangle group. We show that the equation coincides precisely with the one obtained by~Maier~\cite{Maier}. Our derivation has the advantage of being completely elementary without the need for advanced computer algebra or explicit Fourier expansions of the automorphic forms and turns out to follow simply from the existence of the Ramanujan identities and the ring relations.

Before we proceed further we first introduce the notion of a~covariant derivative. In fact, the structure of the Ramanujan identities lends itself to a~natural choice of derivation on ring of automorphic forms. Let the space of weight-$ k $ automorphic forms associated to a~triangle group~$ \ttype $ be denoted~$ \mathfrak{m}_{k} $.\footnote{The dependence on the choice of triangle group is left implicit to avoid clutter.} We define the operator $ \text{D} $ such that
\begin{gather*}
\text{D}\colon\ \mathfrak{m}_{k} \longrightarrow \mathfrak{m}_{k+2},
\end{gather*}
Explicitly, this Ramanujan--Serre derivative is defined as
\begin{gather*}
\text{D} = \frac{1}{2\pi\ii} \frac{\dd}{\dd\tau} - \frac{k}{2} \left(1-\frac{1}{m_1}-\frac{1}{m_2}\right)E_{2,\ttype} .
\end{gather*}
With this definition, the first few Ramanujan identities may be written in~terms of the Rama\-nu\-jan--Serre derivative as
\begin{gather}\label{DE4}
\text{D}E_{4, \ttype} = \frac{2-m_1}{m_1}E_{6, \ttype}^{(1)} + \frac{2-m_2}{m_2} E_{6, \ttype}^{(2)},\\
\label{DE61}
\text{D}E_{6, \ttype}^{(1)} = \frac{3-m_1}{m_1} E_{8, \ttype}^{(1)} + \frac{3-2m_2}{m_2} E_{4,\ttype}^2,\\
\label{DE62}
\text{D}E_{6, \ttype}^{(2)} = \frac{3-m_2}{m_2} E_{8, \ttype}^{(2)} + \frac{3-2m_1}{m_1} E_{4, \ttype}^2.
\end{gather}
We now act with the Serre derivative on \eqref{DE4}, and use equations \eqref{DE61} and \eqref{DE62} along with the ring relation \eqref{eq:notFreelyGenerated}, i.e., we use
\begin{gather*}
E_{4, \ttype}^{(\ell)} E_{8, \ttype}^{(\ell)} = \big(E_{6, \ttype}^{(\ell)}\big)^2 \qquad\text{for}\ \ell=1,2,
\end{gather*}
to obtain the following differential equation relating Eisenstein series' of weights four and six:
\begin{gather}
E_{4, \ttype} \text{D}^2 E_{4, \ttype} = \frac{(3-2m_1)(2-m_2)+(3-2m_2)(2-m_1) }{m_1m_2}E_{4, \ttype}^3 \nonumber
\\ \hphantom{E_{4, \ttype} \text{D}^2 E_{4, \ttype} =}
{}+\frac{(m_1-2)(m_1-3)}{m_1^2} \big(E_{6, \ttype}^{(1)}\big)^2 + \frac{(m_2-2)(m_2-3)}{m_2^2} \big(E_{6, \ttype}^{(2)}\big)^2.\label{E4DDE4eqn}
\end{gather}

Note that each term in~the above equation is automorphic with weight $ 12 $, and consequently that the equation is invariant under triangle group action.

Before turning to the main goal of this section, i.e., deriving the Maier equations for arbitrary triangle groups~\cite{Maier}, we first consider two special cases. These cases are distinguished by simpler Maier equations with lower weight than one might expect, chiefly due to the simplifications that arise when the integers specifying the type~$ \ttype $ of the triangle group are tuned to special values.

\begin{Theorem}[isosceles triangle groups]The Maier equations corresponding to isosceles triangle groups with type $\ttype_M = (M,M,\infty) $ are identical to the Maier equations for the Hecke group~$\mathrm{H}(M) $.
\end{Theorem}

\begin{proof}Equations \eqref{DE4} and \eqref{E4DDE4eqn} simplify considerably for isosceles triangle groups with
\linebreak \mbox{$ m_1 = m_2 = M $}. Using the square of~\eqref{DE4} and the ring relation \eqref{newringrelation}, i.e.,
\begin{gather}\label{E4E6ringrelation}
E_{6, \ttype_{M} }^{(1)}E_{6, \ttype_{M} }^{(2)} = E_{4, \ttype_{M} }^3,
\end{gather}
we can write \eqref{E4DDE4eqn} as
\begin{gather*}
E_{4, \ttype_{M} } \text{D}^2 E_{4, \ttype_{M} } = \frac{M-3}{M-2} (\text{D}E_{4, \ttype_{M} })^2 + \frac{2(M-2)}{M} E_{4, \ttype_{M} }^3.
\end{gather*}
Up to an~unimportant rescaling of~$E_{4, \ttype_{M} }$, this is the same equation that was derived for the Hecke group ${\rm H}(M) $ in~\cite{Ashok:2018myo}.
\end{proof}

This fact was already observed in~\cite{Maier} and we see that it is a~simple consequence of the Ramanujan relations and the ring relations. There exists a~simple geometrical basis for this fact: an~isosceles hyperbolic triangle of type $ (M,M,\infty) $ can be bisected into two hyperbolic triangles of type $ (2,M,\infty) $~\cite{nehari2012conformal}. We now turn to triangles groups of type $ \ttype_{3,M}= (3, M,\infty) $. We present this example because it is instructive and will help build intuition for the most general case.

\begin{Example}[{type $(3,M,\infty)$ triangle groups}]
In this case as well the goal is to write an~equation purely in~terms of the modular covariant derivatives $\text{D}^k E_{4,\ttype_{3,M}}$. Start with \eqref{DE4} and solve for $E_{6, \ttype_{3,M}}^{(1)}$ as
\begin{gather*}
E_{6, \ttype_{3,M}}^{(1)} = -3 \text{D}E_{4, \ttype_{3,M}}-\frac{3(M-2)}{M} E_{6, \ttype_{3,M}}^{(2)}.
\end{gather*}
By combining this with the ring relation \eqref{E4E6ringrelation}, we obtain the relation:
\begin{gather*}
\text{D} E_{4,\ttype_{3,M}} E_{6, \ttype_{3,M}}^{(2)} = -\frac{M-2}{M} \big(E_{6, \ttype_{3,M}}^{(2)}\big)^2- \frac{1}{3} E_{4, \ttype_{3,M}}^3 .
\end{gather*}
One can write an~independent expression for the square of~$E_{6, \ttype_{3,M}}^{(2)}$ in~terms of the $E_{4,\ttype_{3,M}}$ and its modular covariant derivatives courtesy of~\eqref{E4DDE4eqn}:
\begin{gather}\label{E6square}
\big(E_{6, \ttype_{3,M}}^{(2)}\big)^2 = \frac{M^2}{(M-3)(M-2)} E_{4, \ttype_{3,M}} \text{D}^2
E_{4, \ttype_{3,M}} +\frac{M(5M-9)}{3(M-3)(M-2)} (E_{4, \ttype_{3,M}})^3.
\end{gather}
This leads to the following relation:
\begin{gather*}
\text{D} E_{4, \ttype_{3,M}} E_{6, \ttype_{3,M}}^{(2)} = \frac{2}{3} \frac{(2M-3)}{(M-3)}(E_{4, \ttype_{3,M}})^3 -\frac{M}{M-3}E_{4, \ttype_{3,M}} \text{D}^2 E_{4, \ttype_{3,M}}.
\end{gather*}
Squaring both sides of this equation and using the ring relation \eqref{E6square} once again leads to the following equation, with each term of weight 20:
\begin{gather*}
9M^2(M-3) \text{D}^2E_{4, \ttype_{3,M}} (\text{D}E_{4, \ttype_{3,M}})^2 - 9 M^2(M-2) (\text{D}^2E_{4, \ttype_{3,M}})^2 E_{4, \ttype_{3,M}}\\ \qquad
{}+ 3 M(M-3) (9 - 5 M) (\text{D}E_{4, \ttype_{3,M}})^2 (E_{4, \ttype_{3,M}})^2
\\ \qquad
{}+ 12 M (M-2)(2M-3) \text{D}^2E_{4, \ttype_{3,M}} (E_{4, \ttype_{3,M}})^3 -
4(M-2) (2M-3)^2 (E_{4, \ttype_{3,M}})^5 = 0 .
\end{gather*}
\end{Example}

We now bring this strategy to bear on arbitrary triangle groups. This theorem corresponds directly to \cite[Theorem 6.4]{Maier} and we use the same notation for the coefficients for ease of comparison.

\begin{Theorem}For an~arbitrary triangle group of type $\ttype=(m_1,m_2,\infty)$, the weight four automorphic form $E_{4,\ttype}$ satisfies the following weight-$24$ differential equation:
\begin{gather}
C_{88} E_{4, \ttype}^2 \big(\mathrm{D}^2 E_{4, \ttype}\big)^2 + C_{86}E_{4, \ttype}\big(\mathrm{D}E_{4, \ttype}\big)^2\mathrm{D}^2 E_{4, \ttype}+C_{84}E_{4, \ttype}^4\big(\mathrm{D}^2E_{4, \ttype}\big) + C_{66}(\mathrm{D}E_{4, \ttype} )^4\nonumber
\\ \hphantom{C_{88} E_{4, \ttype}^2 \big(\mathrm{D}^2 E_{4, \ttype}\big)^2}
{} + C_{64}E_{4, \ttype}^3(\mathrm{D}E_{4, \ttype} )^2 + C_{44}E_{4, \ttype}^6 = 0, \label{MaierTheorem}
\end{gather}
where the coefficients are given by
\begin{gather}
C_{88} = m_1^2 m_2^2(m_1-2) (m_2-2),\nonumber\\
C_{86} =- m_1^2 m_2^2 ((m_1-2)(m_2-3)+(m_2-2)(m_1-3)),\nonumber\\
C_{84}= -4m_1m_2 (m_1-2) (m_2-2) (m_1m_2 - m_1 -m_2),\nonumber\\
C_{66} =m_1^2 m_2^2(m_1-3) (m_2-3),\nonumber\\
C_{64} = m_1 m_2 \left(3 (m_1-m_2)^2+2 m_2 (m_2-3) (m_1-2)^2+2 m_1 (m_1-3) (m_2-2)^2\right)\!,\nonumber\\
C_{44} = 4 (m_1-2) (m_2-2) (m_1m_2-m_1 -m_2)^2 .\label{MaierCoefficients}
\end{gather}
\end{Theorem}

\begin{proof}First, solve for $E_{6, \ttype}^{(1)}$ using \eqref{DE4}:
\begin{gather*}
E_{6, \ttype}^{(1)} = -\frac{\text{D}E_{4, \ttype} m_1}{m_1-2}-\frac{m_1\left(m_2-2\right)}{m_2\left(m_1-2\right)}E_{6, \ttype}^{(2)}.
\end{gather*}
Then substitute this into the ring relation \eqref{E4E6ringrelation} to obtain the following equation:
\begin{gather}\label{RR1weight12}
\frac{m_2(m_1-2)}{m_1 (m_2-2)}E_{4, \ttype}^3 +\frac{m_2}{m_2-2}E_{6, \ttype}^{(2)} \text{D}E_{4, \ttype} +\big(E_{6, \ttype}^{(2)}\big)^2 = 0.
\end{gather}
A second independent weight-12 equation is arrived at by acting with $E_{4, \ttype}\text{D}$ on \eqref{DE4} and in~the resulting equation, we substitute the equations \eqref{DE61} and \eqref{DE62}. After this, on taking suitable linear combinations we eliminate $\big(E_{6, \ttype}^{(2)}\big)^2$ and obtain the following equation:
\begin{gather}
\text{D}E_{4, \ttype} E_{6, \ttype}^{(2)} = -\frac{(m_1-2) m_2}{m_1-m_2}E_{4, \ttype} \text{D}^2E_{4, \ttype} +\frac{(m_1-3) m_2}{m_1-m_2}(\text{D}E_{4, \ttype})^2\nonumber
\\
\hphantom{\text{D}E_{4, \ttype} E_{6, \ttype}^{(2)} =}
{}+\frac{2(m_1-2) (m_1 m_2 - m_1 - m_2)}{m_1 (m_1-m_2)}E_{4, \ttype}^3 .
\label{DE4E6}
\end{gather}
We substitute this into \eqref{RR1weight12} to obtain the ring relation for $\big(E_{6, \ttype}^{(2)}\big)^2$:
\begin{gather}
\big(E_{6, \ttype}^{(2)}\big)^2 = -\frac{ (m_1-2) m_2^2}{(m_1-m_2) (m_2-2)}E_{4, \ttype}
\text{D}^2E_{4, \ttype}+\frac{ (m_1-3) m_2^2}{(m_1-m_2) (m_2-2)}\text{D}E_{4, \ttype}^2\nonumber
\\ \hphantom{\big(E_{6, \ttype}^{(2)}\big)^2 =}
{}+\frac{ m_2 (m_1-2) (2m_1 m_2 - 3m_1-m_2)}{m_1 (m_1-m_2) (m_2-2)}E_{4, \ttype}^3.
\label{E6sqgeneral}
\end{gather}
We now have all the ingredients in~place to derive the Maier equation for the general triangle group. Consider the square of~\eqref{DE4E6} and substitute the expression for the square of~$E_{6, \ttype}^{(2)}$ derived in~\eqref{E6sqgeneral} above. This leads to the weight-$ 24 $ equation written purely in~terms of~$E_{4, \ttype}$ and its modular covariant derivatives, given in~\eqref{MaierTheorem}, with the coefficients given in~\eqref{MaierCoefficients}. Finally, that the orbit of~$ E_{4,\ttype} $ under the corresponding triangle group also solves the nonlinear differential equation under consideration follows from its manifest covariance.
\end{proof}

Up to a~normalization factor, this precisely matches the Maier equation for the triangle group of type $(m_1, m_2, \infty)$ that is presented in~\cite[Theorem~6.4]{Maier}. We see that this differential equation is a~direct consequence of the Ramanujan identities satisfied by the Eisenstein series in~combination with the algebraic ring relations that bind them.

\section{Painlev\'e analysis}
We now turn to a~study of the nonlinear differential equations that we introduced in~the previous section. In particular, we will show that these equations satisfy the Painlev\'e property, i.e., that the only movable singularities of these differential equations are poles. For a~brief review of the stability of differential equations and the Painlev\'e property, we refer the reader to~\cite{Ashok:2018myo}.

\subsection{Chazy equations}
In this section, we outline a~process of differential elimination that may be used to derive higher-order nonlinear differential equations satisfied by $E_{2,\ttype}$. These differential equations are of order~$m_2$, and are natural analogues of the higher-order Chazy equations studied in~\cite{Ashok:2018myo}.\footnote{We remind the reader that without loss of generality, we choose the integers that specify the type $\ttype$ of the triangle group such that $m_1 \leq m_2 $.}

The process of differential elimination is complicated by presence of two distinct branches of Eisenstein series. In order to sidestep this difficulty, we effect a~change of basis in~the space of automorphic forms that is ``diagonal'' in~the sense that only specific linear combinations of the two branches of Eisenstein series appear together in the Ramanujan identities. For tri\-angle group~$ \ttype $ and $ k \geq 3 $, define the following combinations of the two branches of Eisenstein series:
\begin{gather}\label{eq:G2k}
G_{2k,\ttype} = \prod_{p=2}^{k-1} \left(\frac{m_1 - p}{m_1}\right) E_{2k,\ttype}^{(1)} + \prod_{p=2}^{k-1} \left(\frac{m_2 - p}{m_2}\right) E_{2k,\ttype}^{(2)}.
\end{gather}
Since the Eisenstein series' of weights two and four are unique we identify $ G_{4,\ttype} = E_{4,\ttype} $ and $ G_{2,\ttype} = E_{2,\ttype} $ for uniformity of notation. Once this change of basis is effected, we may proceed with straightforward differential elimination. Note that the $ G_{2k,\ttype} $ are not algebraically independent. This is once again a~consequence of the ring relations. The utility of this change of basis is presented in~the following simple example.

\begin{Example}[{type $(4,4,\infty)$ triangle group}]
In terms of the linear combinations defined in~\eqref{eq:G2k}, the Ramanujan identities~\eqref{eq:RamanujanIdGeneral} associated to the triangle group $ (4,4,\infty) $ take the following form:
\begin{gather*}
G_{2, \mathfrak{t}}^{\prime}=\frac{1}{4}\big(G_{2, \mathfrak{t}}^{2}-G_{4, \mathfrak{t}}\big), \\
G_{4, \mathfrak{t}}^{\prime}=  G_{2, \mathfrak{t}} G_{4, \mathfrak{t}}-G_{6, \mathfrak{t}}, \\
G_{6, \mathfrak{t}}^{\prime}=\frac{1}{4}\big(6 G_{2, \mathfrak{t}} G_{6, \mathfrak{t}}-5 G_{4, \mathfrak{t}}^{2}-4 G_{8, \mathfrak{t}}\big), \\
G_{8, \mathfrak{t}}^{\prime}=\frac{1}{2} ( 4 G_{2, \mathfrak{t}} G_{8, \mathfrak{t}}- G_{4,\mathfrak{t}}G_{6,\mathfrak{t}}).
\end{gather*}
\end{Example}

Using a~procedure outlined in~\cite[Section~4]{Ashok:2018myo}, we may use these equations to arrive at a~differential equation $C_{(4,4)}$ satisfied by $ y = G_{2,\ttype}(\tau)$:
\begin{gather}\label{34chazy}
C_{(4,4)}\colon\ y^{(4)}-5 y y^{(3)}-9 y y'^2+(6 y'+6 y^2) y'' = 0.
\end{gather}

\begin{Remark}
The differential equation \eqref{34chazy}, as we will see shortly, has Bureau symbol P$ 1 $, meaning the differential equation admits solutions having a~Laurent expansion with leading divergence
${\sim}(\tau-\tau_{0})^{-1}$. This equation is known to possess the Painlev\'e property, and appears in~\cite[equation~(2.3)]{Cosgrove2006}.
\end{Remark}
These differential equations are natural analogues of the Chazy equation associated to the quasimodular form of~${\rm PSL}(2,\mathbb{Z}) $~\cite{CLARKSON1996225, Takhtajan:1992qb}, and also the higher-order Chazy eqautions associated to Hecke groups~\cite{Ashok:2018myo}. The above systematic procedure for constructing higher-order nonlinear differential equations motivates the following proposition.

\begin{Proposition}
Each triangle group is associated to a~nonlinear differential equation of order~$m_2$, which we shall denote by $C_{\ttype}$. Further, each term in~this differential equation is of order $2m_2 + 2$.
\end{Proposition}
\begin{proof}\looseness=1 The Ramanujan identities in~terms of the combinations $G_{2k,\ttype}$ are a~system of~$m_2$ first order differential equations in~$m_2$ variables. The proposition follows by differential elimination.
\end{proof}

\subsection{Painlev\'e analysis}
We now briefly recall the main steps of the generalized Painlev\'e analysis, following the work of~\cite{Conte:1993pp}. Consider a~differential equation represented by
\begin{gather}\label{Maiercompact}
M[y] = 0 .
\end{gather}
We first propose a~local solution $y_{c}$ (around a~simple pole singularity) to the nonlinear differential equation in~Frobenius form. One then linearises the equation around the solution
\begin{gather}
\label{linearweqn}
\frac{\dd}{\dd\epsilon}M [y_c + \epsilon w]\Big|_{\epsilon=0} = 0 ,
\end{gather}
and look for solutions $w = w(\tau)$ also in~the Frobenius form. In this case, as in~the original Chazy equation and in~the higher order generalizations in~\cite{Ashok:2018myo}, we find that the linearised solution has poles of higher order than the originally proposed solution to the differential equation. At first glance these so-called ``negative resonances'' appear to make the leading order solution unstable to linear perturbations. However, as shown in~\cite{Conte:1993pp, Fordy:1991nd}, it is possible to perform the stability analysis in~such a~situation by writing the solution as a~functional Frobenius series:
\begin{gather*}
y = \sum_{i=0}^{\infty}y_i \chi^{i-a}.
\end{gather*}
Here, $\chi = \chi(\tau)$ satisfies the Riccati equation of the form (see~\cite{Ashok:2018myo} for more details):
\begin{gather*}
\chi' = 1-\frac{S}{2}\chi^2,
\end{gather*}
and $ y_i $ are the undetermined coefficients in~the Frobenius ansatz. We find that leading order ansatz satisfies the differential equation provided $a$ is an integer solution to an indicial equation.\footnote{Integrality of solutions to an~indicial equation is one of the criteria laid out in~\cite{Conte:1993pp}, and ensures single-valuedness of solutions.} Once a~particular branch of the indicial equation is chosen, we can solve for the coefficients $ y_i $. We then propose the following ansatz for the next to leading order solution:
\begin{gather*}
w = \sum_{i=0}^{\infty} w_i \chi^{i-b} .
\end{gather*}
We then substitute this into the linearized equation in~\eqref{linearweqn} and solve for $b$. The phenomenon of ``negative resonances'' is when $ b $ is a~positive integer larger than $ a~$. This appears to destabilize the leading order solution as it appears to have a~more singular behaviour near $\chi=0$. However on setting the coefficients of the leading order solution to be the $ y_i $ that we have solved for, we find that the coefficients of the higher order negative resonance vanishes and this implies that the negative resonance does not destabilize the solution of the differential equation. Thus, following the methods of~\cite{Conte:1993pp, Fordy:1991nd} we can claim that the Chazy equation for the triangle groups of type $ (M,M,\infty) $ satisfies the Painlev\'e property and this leads us to the following proposition.

\begin{Proposition}\label{prop42}
For a~triangle group of type $\ttype=(M,M,\infty)$, the associated Chazy equation~$C_{\ttype}$ possesses the Painlev\'e property. The Eisenstein series $y=E_{2,\ttype}(\tau)$ that satisfies both the differential equations has the following leading order expansion:
\begin{gather*}
y = \frac{M}{M-2}\left(-\frac{2}{\chi} + \frac{S}{3} \chi + \frac{2S^2}{45}\chi^3 + \frac{11S^3}{945}\chi^5 +\cdots \right)\!.
\end{gather*}
The indicial equation for the resonances has the roots $ b\in \{2, \dots, m+1\}$.
\end{Proposition}

This extends the results obtained for the Hecke groups in~\cite{Ashok:2018myo} to all isosceles triangle groups with a~cusp. Before concluding, we briefly discuss the status of the Chazy equations for triangle groups of type $ (m_1,m_2,\infty) $ with $ m_1 \neq m_2 $, and the Maier equation presented in Theorem~\ref{MaierTheorem}. In the former case, some of the resonance numbers $ b $ are negative rational numbers. We have observed that triangle groups of type $ (m,m+k,\infty) $ have $ k $ such negative rational resonances. For the Maier equation, the analysis is more intricate~\cite{pickering2003painleve}.\footnote{On revisiting our previous work, it appears that the our analysis of the Maier equation \cite[equation~(4.5)]{Ashok:2018myo} using the methods of Conte--Fordy--Pickering~\cite{Conte:1993pp} are insufficient to prove the Painlev\'e property. The conclusions of \cite[Proposition~5.1]{Ashok:2018myo} pertaining to the higher-order Chazy equations $ C_m $ \cite[Proposition~4.1]{Ashok:2018myo}, however, remain unchanged. A more comprehensive analysis of the Maier equation is in~preparation.} In both these cases, the criteria of~\cite{Conte:1993pp,Fordy:1991nd} that test for the Painlev\'e property are not straightforwardly applicable. Work on extending these techniques is currently underway.

\appendix

\section{Generalized Halphen system}\label{app:Fourier}

In this section we collect a~few results about the Halphen system and the Fourier expansions of the automorphic forms $E_{2k,\ttype}$ associated to triangle group of type $\ttype = (m_1, m_2, \infty)$. In terms of the parameters $ (a,b,c) $ introduced in~\eqref{eq:abcHalphen}, the explicit solution of the Halphen system in~\eqref{eq:GenHalphen} is given by (see \cite[Theorem 3]{Doran:2013npa}):
\begin{gather}
t_{1,\ttype}(\tau) =\frac{1}{\alpha_\ttype} (a-1) z Q(a, b; z)\, {}_2F_1(1-a, b;1 ;z ) {}_2F_1(2-a, b;2 ; z),\nonumber\\
t_{2,\ttype}(\tau) - t_{1,\ttype}(\tau) = \frac{1}{\alpha_\ttype} Q(a, b; z)\, {}_2F_1(1-a, b;1 ;z )^2,\nonumber\\
t_{3,\ttype}(\tau) -t_{1,\ttype}(\tau) = \frac{1}{\alpha_\ttype} z Q(a, b; z)\, {}_2F_1(1-a, b;1 ;z )^2.\label{eq:HalphenGenSol}
\end{gather}
The function ${}_2F_1$ is the Gauss hypergeometric function while the function $Q(a ,b; z)$ is given by
\begin{gather*}
Q(a, b; z) = \frac{\ii \pi (1-b)}{2\sin (\pi b) \sin (\pi a)} (1-z)^{b-a} .
\end{gather*}
The parameter $z$ is related to the standard hauptmodul $J_{\ttype}$ of the triangle group:
\begin{gather*}
z = \frac{1}{1-J_{\ttype}} .
\end{gather*}
The solutions in~\cite{Doran:2013npa} differ from those in~\eqref{eq:HalphenGenSol} by an~overall constant $\alpha_\ttype$ that we are free to choose. We choose this coefficient such that the expansion for $t_{2,\ttype}$ begins with unit coefficient; this uniquely fixes the coefficient $\alpha_\ttype$ to be
\begin{gather*}
\alpha_{\ttype} = -\frac{m_1+m_2+m_1m_2}{8m_1m_2}\csc\left[\frac{\pi}{2}\left(1-\frac{1}{m_1}+\frac{1}{m_2}\right)\right] \sec\left[\frac{\pi(m_1+m_2)}{2m_1m_2} \right]\! .
\end{gather*}
With this choice of normalisation, the quasi-automorphic weight-$ 2 $ Eisenstein series $ E_{2,\ttype} $ has the following behavior as $ \tau \rightarrow \ii\infty $:
\begin{gather*}
E_{2,\ttype}(\tau) = 1 + O(q).
\end{gather*}

\section{Chazy equations}\label{app:Chazy}
In this appendix we give examples of nonlinear differential equations naturally associated to tri\-angle groups and solved by their quasiautomorphic weight-$ 2 $ Eisenstein series. These dif\-fe\-ren\-tial equations were discussed in~Proposition~\ref{prop42} and are arrived at by differential elimination on the system of equations governing the $ G_{2k,\ttype} $ defined in~\eqref{eq:G2k}. For a~triangle group of type $ \ttype = (m_1, m_2, \infty) $, they are denoted $ C_{\ttype} \equiv C_{(m_1,m_2)} $. The following equations all possess the Painlev\'e property:
\begin{gather*}
C_{(3,3)}\colon\ y^{(3)}-2 y y''+3 y'^2 = 0,
\\
C_{(4,4)}\colon\ y^{(4)}-5 y y^{(3)}-9 y y'^2+\big(6 y'+6 y^2\big) y'' = 0,
\\
C_{(5,5)}\colon\ 75 y^{(5)}-675 y y^{(4)}+520 y''^2-366 y'^3+2727 y^2 y'^2
\\ \phantom{C_{(5,5)}\colon\ }
{}+y^{(3)} \big(220 y'+1959 y^2\big)+\big({-}2664 y y'-1818 y^3\big) y'' = 0,
\\
C_{(6,6)}\colon\ 3 y^{(6)}-42 y^{(5)} y-228 y y''^2-552 y^3 y'^2+228 y y'^3+y^{(4)} \big(214 y^2-12 y'\big)
\\ \phantom{C_{(6,6)}\colon\ }
{}+\big(600 y^2 y'-114 y'^2+368 y^4\big) y''+y^{(3)} \big(57 y''-34 y y'-468 y^3\big) = 0.
\end{gather*}

\subsection*{Acknowledgments}
We thank Robert Conte for helpful correspondence, and Hossein Movasati for valuable discussions. We would also like to thank the anonymous referees for valuable comments and feedback. DPJ and MR are grateful to IMSc, Chennai for hospitality. MR acknowledges support from the Infosys Endowment for Research into the Quantum Structure of Spacetime. DPJ acknowledges support from SERB grant CRG/2018/002835.

\providecommand{\eprint}[2][]{\href{http://arxiv.org/abs/#2}{arXiv:#2}}

\pdfbookmark[1]{References}{ref}
\LastPageEnding
\end{document}